\documentclass[11pt]{article}

\usepackage[T1]{fontenc}
\usepackage[utf8]{inputenc}
\usepackage{lmodern}
\usepackage{amsmath,amssymb,amsthm,bm}
\usepackage{graphicx}
\usepackage{microtype}
\usepackage{geometry}
\usepackage{authblk}
\usepackage{hyperref}
\usepackage[numbers,sort&compress]{natbib}
\usepackage{booktabs}
\usepackage{enumitem}
\usepackage{caption}
\usepackage{subcaption}
\usepackage{float}
\newtheorem{proposition}{Proposition}

\newcommand{\Ttask}{\mathcal{T}}
\newcommand{\Dtask}{\mathcal{D}}
\newcommand{\Ptask}{\mathcal{P}}
\newcommand{\Etask}{\mathcal{E}}
\newcommand{\Rtask}{\mathcal{R}}
\newcommand{\Stask}{\mathcal{S}}
\newcommand{\Otask}{\mathcal{O}}

\newcommand{\dd}{\mathrm{d}}

\title{Constructor-Theoretic Optical Time:\\Delay, Phase, Fisher Distinguishability, and Double-Slit Records as Physical Tasks}

\author[1]{Juan Sumaya-Mart\'{i}nez\thanks{Corresponding author: j.sumaya2011@gmail.com}}
\author[2]{Omar Olmos-L\'{o}pez}
\affil[1]{Faculty of Sciences, Universidad Aut\'{o}noma del Estado de M\'{e}xico, Toluca 50000, M\'{e}xico}
\affil[2]{School of Engineering and Sciences, Tecnol\'{o}gico de Monterrey, Monterrey, Nuevo Le\'{o}n, M\'{e}xico}

\date{}

\begin{document}
\maketitle

\begin{abstract}
Constructor theory proposes that physical laws may be expressed in terms of possible and impossible tasks rather than in terms of primitive time-parametrized evolution. This raises a concrete question for optical physics, where delay, phase, synchronization, time of flight, coherence time, and detector records are normally described using fields of the form $E(\mathbf r,t)$. We develop a constructor-theoretic formulation of optical time in which delay, phase, temporal ordering, synchronization, and record formation are treated as physical tasks acting on optical, reference, and record substrates. An optical delay is not introduced as a primitive external parameter; it is an attribute made operationally meaningful by the possibility of transformations that map input optical attributes into distinguishable output records. We define delay constructors, phase-delay equivalence tasks, temporal-ordering tasks, synchronization tasks, and Fisher-admissible temporal-estimation tasks. Two formal propositions state that optical temporal attributes are operationally defined only relative to comparison and record-forming tasks, and that the Cramer-Rao bound is a constructor-theoretic impossibility boundary for a specified measurement class. The framework is illustrated with an interferometric delay constructor, a dispersive group-delay constructor, and a double-slit diffraction constructor. In the double-slit case the usual diffraction pattern is recovered as the record distribution of a phase-delay comparison task, with path delay $\Delta\tau(\theta)=d\sin\theta/c$ and phase $\phi(\theta)=\omega\Delta\tau(\theta)$. We also compare this task-based record distribution with the electromagnetic Fraunhofer aperture result and show that they coincide in the common scalar, paraxial, thin-aperture limit. The proposal does not replace Maxwellian optics; rather, it reorganizes optical time as an emergent operational structure arising from possible transformations, distinguishable records, and resource-dependent impossibility statements.
\end{abstract}

\noindent\textbf{Keywords:} constructor theory; time; optical delay; phase; Fisher information; double-slit diffraction; foundations of physics; optical metrology; distinguishability.

\section{Introduction}
The formalism of optical physics is saturated with temporal notions. A quasi-monochromatic field is commonly written as $E(\mathbf r,t)=\mathrm{Re}\{\mathcal E(\mathbf r)e^{-i\omega t}\}$; a delay line is represented by $E(t)\mapsto E(t-\tau)$; a phase shifter implements $\psi\mapsto e^{i\phi}\psi$; a dispersive medium is described by a spectral phase $k(\omega)L$; and an interferometer converts a path delay into a measurable intensity modulation. These formulations are indispensable in practical optics. Yet they presuppose a temporal parameter $t$ against which propagation, phase accumulation, and measurement are described.

The status of time in fundamental physics is more subtle. Page and Wootters described apparent dynamics through correlations with internal clock degrees of freedom rather than through an external time parameter \citep{PageWootters1983}. Rovelli argued that background-independent physics motivates a relational formulation in which time is not part of the fundamental vocabulary \citep{Rovelli2011ForgetTime}. Salecker and Wigner emphasized that measurements of spacetime intervals have physical limitations imposed by the quantum nature of measuring devices \citep{SaleckerWigner1958}. Deutsch's constructor theory proposes a different explanatory mode: laws are to be expressed as statements about which transformations of physical substrates are possible or impossible, and why \citep{Deutsch2013Constructor}. In the constructor theory of information, information is not an a priori mathematical object, but a physically instantiated set of possible and impossible transformations \citep{DeutschMarletto2015Info}. Most recently, the constructor theory of time asks how duration and dynamics can be meaningful when laws in their constructor-theoretic form do not refer to time \citep{DeutschMarlettoTime2025}. Experimental directions for testing constructor-theoretic principles are also beginning to be discussed \citep{MarlettoDeutschVedral2026Tests}.

This paper asks a restricted but concrete question: what becomes of optical time if the fundamental statement is not ``the field evolves in time'' but ``a certain optical task is possible or impossible''? We do not claim to derive the full temporal structure of physics, nor do we deny the usefulness of $E(t)$. We propose a domain-specific operational reconstruction: optical time is what is made available by tasks that delay, phase-shift, order, synchronize, compare, and record optical attributes.

The central claim is that delay and phase become temporal only through a network of transformations. A delay $\tau$ is not directly visible. It becomes a physical temporal attribute when a constructor transforms it into records distinguishable from the records generated by other delays. A phase $\phi$ becomes a temporal attribute only relative to a frequency reference and a phase-delay equivalence task. A clock is not an external parameter; it is a reference substrate participating in synchronization and ordering tasks. A measured temporal fact requires stable record formation.

Fisher information then enters naturally. If a delay $\tau$ controls a probability distribution $p(y|\tau)$ over measurement records, the classical Fisher information
\begin{equation}
F(\tau)=\int \dd y\,p(y|\tau)\left[\frac{\partial}{\partial\tau}\ln p(y|\tau)\right]^2
\label{eq:fisher_general}
\end{equation}
quantifies the local distinguishability of nearby delays. In the task-based reading proposed here, $F(\tau)$ is a resource that determines whether the task
\begin{equation}
\Etask_{\tau,\Delta\tau}:
\hbox{unknown delay attribute}\longrightarrow \hbox{classical estimate }\hat\tau
\end{equation}
can be realized with target root-mean-square uncertainty $\Delta\tau$. Under the usual regularity conditions for an unbiased estimator, the Cramer-Rao bound,
\begin{equation}
\mathrm{Var}(\hat\tau)\geq F(\tau)^{-1},
\label{eq:cramer_rao_intro}
\end{equation}
becomes a task-impossibility statement: demanding $\Delta\tau^2<F(\tau)^{-1}$ specifies a transformation that is impossible for the selected optical substrate, reference, detector, measurement architecture, and resource budget. This is consistent with the role of Fisher information in optical metrology and resolution theory \citep{Fisher1922,Cramer1946,Rao1945,ChaoWardOber2016,Tsang2016}.

The manuscript is organized as follows. Section~\ref{sec:ct_background} introduces the constructor-theoretic vocabulary and clarifies the non-temporal reading of repeatability. Section~\ref{sec:substrates} defines optical, reference, and record substrates. Sections~\ref{sec:delay}--\ref{sec:ordering} formulate delay, phase-delay equivalence, temporal ordering, and synchronization as tasks. Section~\ref{sec:fisher} gives the Fisher-information possibility boundary and states the two main propositions. Sections~\ref{sec:interferometer}, \ref{sec:dispersion}, and \ref{sec:double_slit} illustrate the framework with interferometric delay estimation, dispersive group delay, and double-slit diffraction, including an explicit comparison with the electromagnetic Fraunhofer aperture result. Section~\ref{sec:records} discusses records and irreversibility. Sections~\ref{sec:relation}, \ref{sec:discussion}, and \ref{sec:conclusion} discuss scope, limitations, and outlook.

\section{Constructor-theoretic background}
\label{sec:ct_background}
Constructor theory replaces the usual dynamical question, ``what happens next?'', by the counterfactual question, ``which transformations can be made to happen, and which cannot?'' \citep{Deutsch2013Constructor}. A physical system capable of possessing attributes is called a substrate. An attribute is a set of microscopic states sharing a specified property. A task is a set of input-output pairs
\begin{equation}
\Ttask=\{x_i\rightarrow y_i\}_{i\in I},
\end{equation}
where $x_i$ and $y_i$ are attributes of one or more substrates. A constructor for a task is a physical system that can cause the transformation while retaining the relevant attributes required to instantiate the same task again to arbitrary accuracy.

That last phrase is potentially delicate in a paper about time. The repeatability of a constructor should not be read here as primitive temporal succession. Rather, it is a counterfactual property: the constructor retains the attributes that make the specified transformation possible across allowed instantiations. In ordinary laboratory language this is described dynamically as repeated operation in time, but the constructor-theoretic statement itself concerns the possibility of the task and the persistence of the relevant constructor attributes, not the fundamental existence of an external temporal parameter.

A task may be possible, impossible, or possible only to an approximation. The explanatory content lies in specifying the boundary between possible and impossible transformations. Constructor theory of information applies this viewpoint to information-bearing media and grounds information in physically possible transformations \citep{DeutschMarletto2015Info}. Constructor theory of time extends the program by asking how duration and dynamics can be meaningful even though laws in constructor-theoretic form do not refer to time \citep{DeutschMarlettoTime2025}.

The present work does not attempt to reconstruct the full constructor theory of time. Instead, it identifies an optical subproblem. In practical photonics, the temporal concepts most frequently encountered are delay, phase, ordering, synchronization, coherence time, and measurement records. Each is already operational. A delay line implements a transformation. A phase shifter implements a transformation. An interferometer maps delay into intensity records. A detector maps field attributes into classical records. Optical time can therefore be studied as a structure arising from tasks among substrates.

\section{Optical substrates, attributes, and records}
\label{sec:substrates}
We introduce three families of substrates.

First, the field substrate $S_F$ supports optical attributes such as spectral amplitude, temporal envelope, phase, polarization, spatial mode, coherence, or photon-number distribution. Second, the reference substrate $S_C$ supports clock-like or phase-reference attributes: a local oscillator, an optical frequency comb, a stable cavity, a synchronized pulse train, or another optical mode used as a reference. Third, the record substrate $S_R$ supports stable detector records: spatial counts, temporal arrival bins, homodyne currents, interferometric output counts, or macroscopic memory states.

A general optical-temporal task may be written as
\begin{equation}
\Ttask:x_F\otimes x_C\otimes x_R\longrightarrow y_F\otimes y_C\otimes y_R,
\end{equation}
where $x_F,x_C,x_R$ are input attributes and $y_F,y_C,y_R$ are output attributes. In many metrological tasks the final field attribute is discarded and the essential output is a stable record,
\begin{equation}
\Ttask_{\rm rec}:x_F\otimes x_C\otimes x_R\longrightarrow r(\theta),
\end{equation}
where $\theta$ is the parameter encoded in the task. In this paper $\theta$ will often be a delay $\tau$, a phase $\phi$, a slit separation $d$, or a detection angle $\theta$.

Figure~\ref{fig:task_network} summarizes the operational structure. A preparation attribute and a reference substrate are processed by an optical constructor. The output record supports inference. The temporal parameter is meaningful because the task network produces distinguishable records, not because an external time variable was added as a primitive explanatory element.

\begin{figure}[t]
\centering
\includegraphics[width=0.92\linewidth]{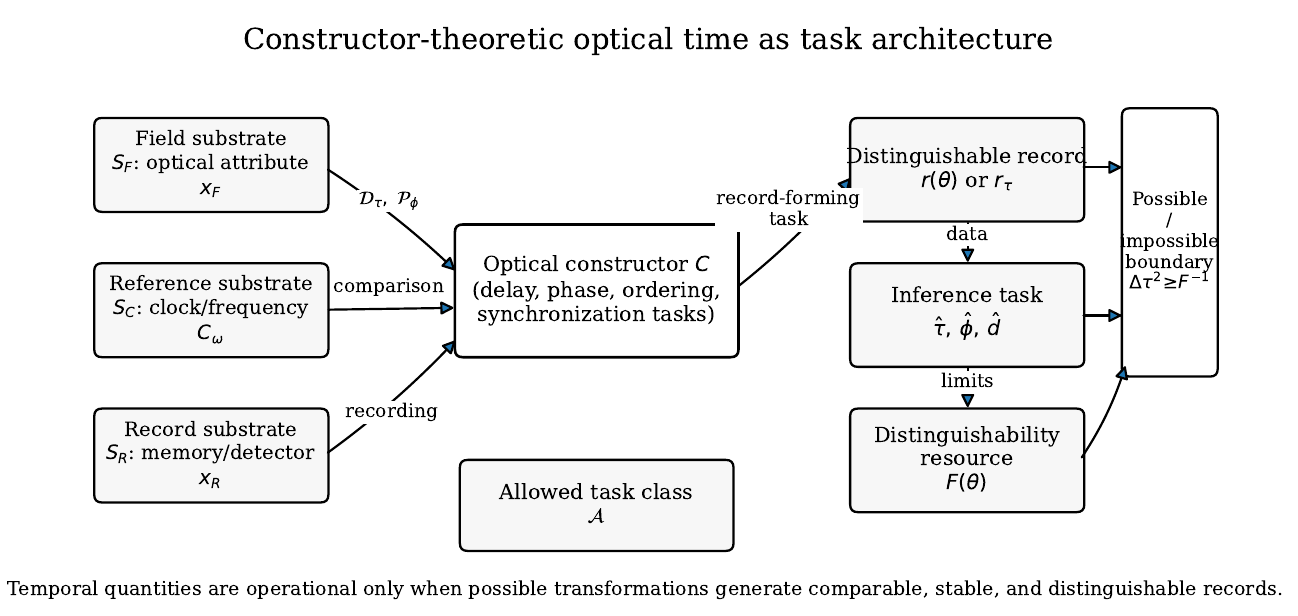}
\caption{Constructor-theoretic optical time as a network of physical tasks. A preparation attribute and a reference substrate are acted on by an optical constructor. The resulting record has a distinguishability resource, quantified locally by Fisher information, which determines whether the inference task is possible to a prescribed accuracy.}
\label{fig:task_network}
\end{figure}

\section{Delay without primitive time}
\label{sec:delay}
In ordinary optics, a delay is represented as
\begin{equation}
\psi(t)\longrightarrow \psi(t-\tau).
\end{equation}
This notation is compact and useful, but it treats $t$ as an already available parameter. We instead define a delay constructor as a device implementing a task
\begin{equation}
\Dtask_\tau:x_F\longrightarrow x_F^{(\tau)},
\label{eq:delay_task}
\end{equation}
where $x_F^{(\tau)}$ is an output attribute empirically equivalent, relative to a specified reference and detector class, to an optical delay $\tau$.

The phrase ``relative to a specified reference and detector class'' is essential. A delay is not an isolated metaphysical property. It is meaningful because comparison tasks can relate the output attribute to a reference. Thus two implementations $\Dtask_\tau^{(1)}$ and $\Dtask_\tau^{(2)}$ belong to the same operational delay class if all allowed comparison tasks produce record distributions indistinguishable within tolerance $\epsilon$:
\begin{equation}
\Dtask_\tau^{(1)}\sim_\epsilon \Dtask_\tau^{(2)}
\quad\Longleftrightarrow\quad
D\big[p_1(r|\tau),p_2(r|\tau)\big]\leq\epsilon,
\label{eq:delay_equiv}
\end{equation}
where $D$ is an operational distance between record distributions. In local estimation problems, the Fisher metric provides a natural infinitesimal distance,
\begin{equation}
\dd s^2=F(\tau)\dd\tau^2.
\end{equation}

A free-space path, an optical fiber, a Fabry-Perot cavity, a slow-light medium, and a programmable pulse shaper may instantiate different microscopic dynamics while realizing approximately the same delay task. What they share is not a primitive time coordinate but a common input-output relation relative to a class of optical comparison tasks.

\begin{figure}[t]
\centering
\includegraphics[width=0.84\linewidth]{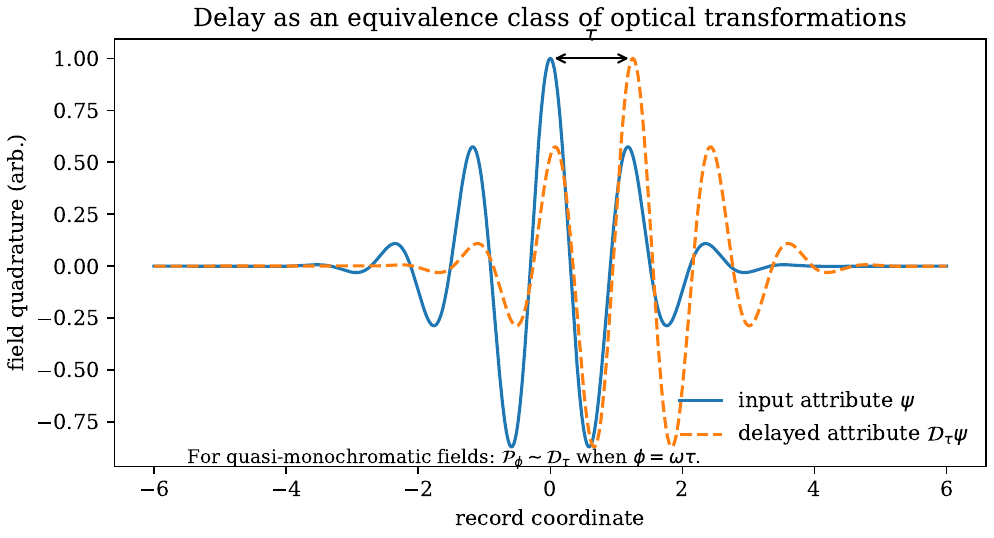}
\caption{Delay as an equivalence class of optical transformations. The output pulse is not interpreted primarily as a field that has evolved under primitive time, but as the result of a delay task $\Dtask_\tau$ that produces a distinguishable output attribute. For quasi-monochromatic fields, phase and delay become operationally equivalent through $\phi=\omega\tau$ only relative to a frequency reference.}
\label{fig:delay_equiv}
\end{figure}

\section{Phase-delay equivalence}
\label{sec:phase_delay}
Optical phase is often the most precise carrier of temporal information. A phase shifter implements
\begin{equation}
\Ptask_\phi:\psi\longrightarrow e^{i\phi}\psi.
\end{equation}
For a monochromatic reference of angular frequency $\omega$, one writes
\begin{equation}
\phi=\omega\tau.
\end{equation}
In the present formulation this equation is not merely a substitution. It defines a phase-delay equivalence task: a phase transformation and a delay transformation are operationally equivalent when they lead to the same class of records under a specified reference frequency.

Let $C_\omega$ denote a reference substrate supporting a sufficiently sharp frequency attribute. We define the phase-delay equivalence relation by
\begin{equation}
\Ptask_\phi\equiv_{C_\omega,\epsilon}\Dtask_\tau
\quad\hbox{if}\quad
D\big[p(r|\Ptask_\phi,C_\omega),p(r|\Dtask_\tau,C_\omega)\big]\leq\epsilon,
\label{eq:phase_delay_equiv}
\end{equation}
with $\phi=\omega\tau$ within the operational tolerance. Without a frequency reference, a phase shift is not yet a time interval; it is only a transformation of an optical attribute. With a reference, the task becomes a delay-estimation task.

This point is important for foundations. Optical time is often treated as if it were simply read off from phase. But phase does not define time by itself. It defines time only within a network of transformations that includes a reference, an interference or comparison operation, and a stable record.

\section{Temporal ordering and synchronization as tasks}
\label{sec:ordering}
A temporal-ordering task concerns two event-like optical attributes $A$ and $B$. It may be written as
\begin{equation}
\Otask:\{A,B\}\longrightarrow\{A\prec B,\; B\prec A,\; A\sim B\},
\end{equation}
where $\prec$ denotes operational precedence and $\sim$ denotes indistinguishability within the available resolution. The task is possible only if the record substrate can distinguish the relevant alternatives. For two pulse-arrival attributes separated by $\delta\tau$, the ordering task is locally impossible when
\begin{equation}
F(\tau)(\delta\tau)^2\ll 1.
\end{equation}

A synchronization task involves two reference substrates $C_1$ and $C_2$:
\begin{equation}
\Stask:C_1\otimes C_2\longrightarrow C_1\sim C_2.
\end{equation}
In optical laboratories this may be implemented by phase-locking lasers, stabilizing frequency combs, distributing a clock over a fiber, or using correlated pulses. In the task-based reading, synchronization is not the assertion that two systems share a primitive external time. It is the possibility of maintaining a relation among reference attributes under a prescribed class of comparison tasks.

This formulation also clarifies why clocks and memories are not optional add-ons. A temporal property must be recordable, comparable, or reproducible to enter physics. A system that cannot support ordering, synchronization, or record-forming tasks cannot function as an optical time substrate.

\section{Fisher distinguishability and temporal-estimation tasks}
\label{sec:fisher}
Consider a delay-dependent measurement with record distribution $p(y|\tau)$. The Fisher information is
\begin{equation}
F(\tau)=\int \dd y\,\frac{1}{p(y|\tau)}\left[\frac{\partial p(y|\tau)}{\partial\tau}\right]^2.
\end{equation}
If the total resource number is $N$ and the records are independent, the total Fisher information scales as $N F_1(\tau)$ under ideal conditions.

We define a temporal-estimation task
\begin{equation}
\Etask_{\tau,\Delta\tau}:x_F^{(\tau)}\otimes x_C\otimes x_R\longrightarrow \hat\tau,
\end{equation}
where the desired output is a classical estimate satisfying
\begin{equation}
\mathrm{Var}(\hat\tau)\leq \Delta\tau^2.
\end{equation}
For unbiased estimators satisfying the usual regularity assumptions, the task is Fisher-admissible when
\begin{equation}
\Delta\tau^2\geq F(\tau)^{-1},
\label{eq:fisher_admissible}
\end{equation}
and Fisher-impossible for the specified measurement class when
\begin{equation}
\Delta\tau^2< F(\tau)^{-1}.
\label{eq:fisher_impossible}
\end{equation}

\begin{proposition}[Operational definability of optical temporal attributes]
An optical temporal attribute, such as a delay, phase-derived time interval, ordering relation, or synchronization relation, is operationally definable only relative to a class of comparison and record-forming tasks. If no admissible task maps the putative temporal attribute into distinguishable records, then the attribute has no operational temporal content within that optical theory.
\end{proposition}

\begin{proof}
Let $\Theta=\{\theta_i\}$ denote a set of putative temporal attributes and let $\mathcal A$ be the allowed class of comparison and record-forming tasks. The operational content of $\theta_i$ is exhausted by the family of record distributions $\{p_A(r|\theta_i):A\in\mathcal A\}$. If for all $A\in\mathcal A$ and all $i,j$ one has $D[p_A(r|\theta_i),p_A(r|\theta_j)]=0$ within the tolerated resolution, then no allowed transformation distinguishes the attributes. They therefore define no distinct temporal property in the operational theory. Conversely, if at least one allowed task produces distinguishable records, the attribute acquires operational content relative to that task class. This is precisely the task-based reconstruction of delay, phase-derived time, ordering, and synchronization used above.
\end{proof}

\begin{proposition}[Fisher-information impossibility boundary]
For a specified optical substrate, reference, detector, noise model, and measurement architecture, the Cramer-Rao inequality defines a constructor-theoretic impossibility boundary for local temporal-estimation tasks: a task demanding $\Delta\tau^2<F(\tau)^{-1}$ is impossible within the specified measurement class.
\end{proposition}

\begin{proof}
The specified substrates and architecture define a family of record distributions $p(y|\tau)$ and hence a Fisher information $F(\tau)$. For any unbiased estimator satisfying the regularity conditions of the Cramer-Rao theorem, $\mathrm{Var}(\hat\tau)\geq F(\tau)^{-1}$. Therefore a task whose output specification requires $\mathrm{Var}(\hat\tau)\leq\Delta\tau^2<F(\tau)^{-1}$ contradicts the statistical distinguishability constraints generated by that task family. Since constructor-theoretic possibility requires the physical realizability of the specified input-output transformation, the requested task is impossible relative to the specified measurement class.
\end{proof}

Equations~\eqref{eq:fisher_admissible} and \eqref{eq:fisher_impossible} are not new statistical facts. Their novelty here is interpretive and organizational: a statistical estimation bound is recast as a physical impossibility statement about transformations. Three kinds of impossibility should be distinguished: structural impossibility, when $\partial_\tau p(y|\tau)=0$ and $F(\tau)=0$; resource-limited impossibility, when the task is impossible for the available photons, bandwidth, visibility, detector response, or reference stability but may become possible with improved resources; and fundamental impossibility, when a task remains impossible even as ordinary resources are increased, for example because of quantum distinguishability limits or constraints imposed by the allowed physical transformations \citep{Helstrom1976,BraunsteinCaves1994,Giovannetti2006,Paris2009}.

\begin{figure}[t]
\centering
\includegraphics[width=0.82\linewidth]{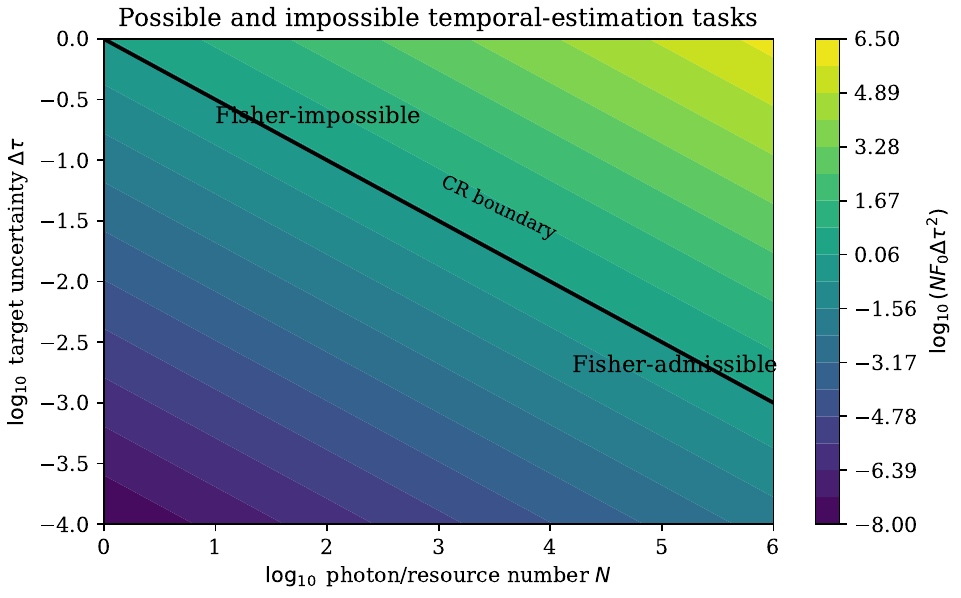}
\caption{Possible and impossible temporal-estimation tasks. For a resource number $N$ and target uncertainty $\Delta\tau$, the Cramer-Rao boundary separates Fisher-admissible tasks from tasks demanding more temporal distinguishability than the specified constructor can provide.}
\label{fig:possible_impossible}
\end{figure}

\section{Interferometric delay constructor}
\label{sec:interferometer}
A two-output interferometer provides the simplest quantitative example. Let a relative delay $\tau$ produce a phase
\begin{equation}
\phi=\omega\tau+\phi_0,
\end{equation}
where $\omega$ is the reference angular frequency and $\phi_0$ is a controllable bias phase. The two output probabilities of a balanced interferometer with visibility $V$ may be written as
\begin{equation}
p_\pm(\phi)=\frac{1}{2}\left[1\pm V\cos\phi\right].
\end{equation}
The Fisher information for estimating $\phi$ from the two-output record is
\begin{equation}
F_\phi=\sum_{s=\pm}\frac{1}{p_s(\phi)}\left[\frac{\partial p_s(\phi)}{\partial\phi}\right]^2
=\frac{V^2\sin^2\phi}{1-V^2\cos^2\phi}.
\end{equation}
For $N$ independent detected photons, the delay Fisher information is
\begin{equation}
F_\tau=N\omega^2F_\phi,
\end{equation}
and therefore
\begin{equation}
\Delta\tau\geq \frac{1}{\omega\sqrt{N F_\phi}}.
\end{equation}

This is normally read as a metrological sensitivity bound. In the present formulation it states that a two-output interferometer is a constructor mapping the delay attribute $\tau$ into a binary record. Only delay-estimation tasks satisfying the bound are possible for the specified visibility and photon budget. Visibility loss reduces the possibility of the task because it reduces the distinguishability of records, as shown in Fig.~\ref{fig:interferometer_fisher}.

\begin{figure}[t]
\centering
\includegraphics[width=0.82\linewidth]{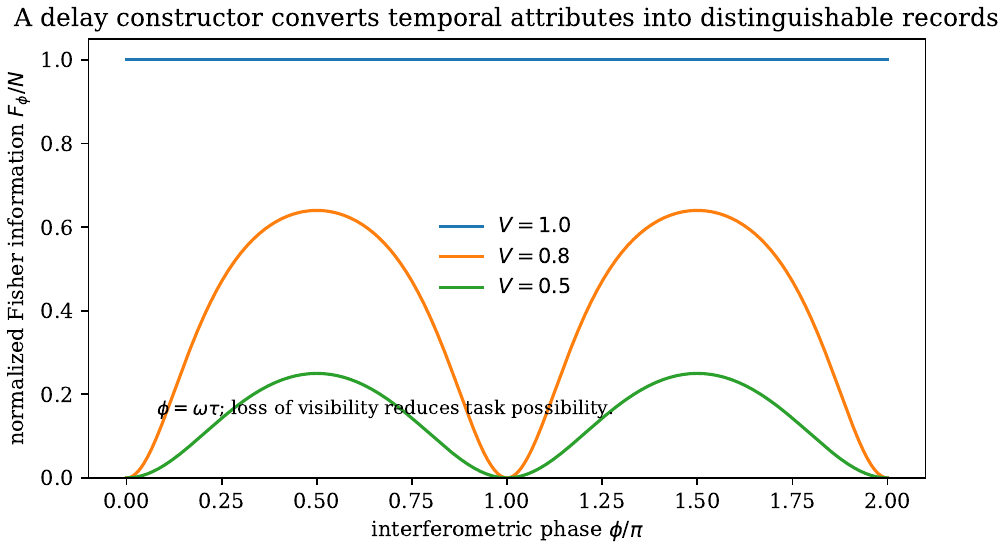}
\caption{Fisher information of a two-output interferometric delay constructor. The interferometer converts a delay into a phase-dependent record. Reduced visibility lowers the Fisher information and therefore enlarges the class of Fisher-impossible delay-estimation tasks.}
\label{fig:interferometer_fisher}
\end{figure}

Several foundational points follow. First, $\tau$ is inferred through a phase-delay equivalence task, not measured as a primitive time interval. Second, the reference frequency $\omega$ is part of the constructor; without it, the phase record does not define a delay. Third, the record is essential: before detection and stable storage, the optical transformation has not yet produced a classical temporal fact.

\section{Dispersive group-delay constructor}
\label{sec:dispersion}
A dispersive medium provides a second example. In frequency space, propagation through a length $L$ is represented by
\begin{equation}
\widetilde E_{\rm out}(\omega)=\widetilde E_{\rm in}(\omega)e^{ik(\omega)L}.
\end{equation}
Expanding around a carrier frequency $\omega_0$,
\begin{equation}
k(\omega)L\simeq k_0L+(\omega-\omega_0)\tau_g+\frac{1}{2}(\omega-\omega_0)^2\beta_2L+\cdots,
\end{equation}
where
\begin{equation}
\tau_g=\left.\frac{\dd}{\dd\omega}[k(\omega)L]\right|_{\omega_0}
\end{equation}
is the group delay and $\beta_2$ is the group-velocity dispersion coefficient. Standard optics interprets this as propagation through a medium \citep{MandelWolf1995,SalehTeich2019}. In the task-based formulation it is a transformation of spectral attributes into output temporal-record attributes,
\begin{equation}
\mathcal G:x_{\rm spectral}\longrightarrow x_{\rm group\ delay}\otimes x_{\rm broadened}.
\end{equation}

The group delay is operationally meaningful only because comparison tasks can relate the output pulse to a reference pulse. When dispersion broadens the pulse, the ordering and delay-estimation tasks may become less precise even if the mean group delay remains well defined. Figure~\ref{fig:dispersion} now displays a visible broadening sequence, emphasizing that the constructor changes the temporal distinguishability structure of the pulse.

\begin{figure}[t]
\centering
\includegraphics[width=0.82\linewidth]{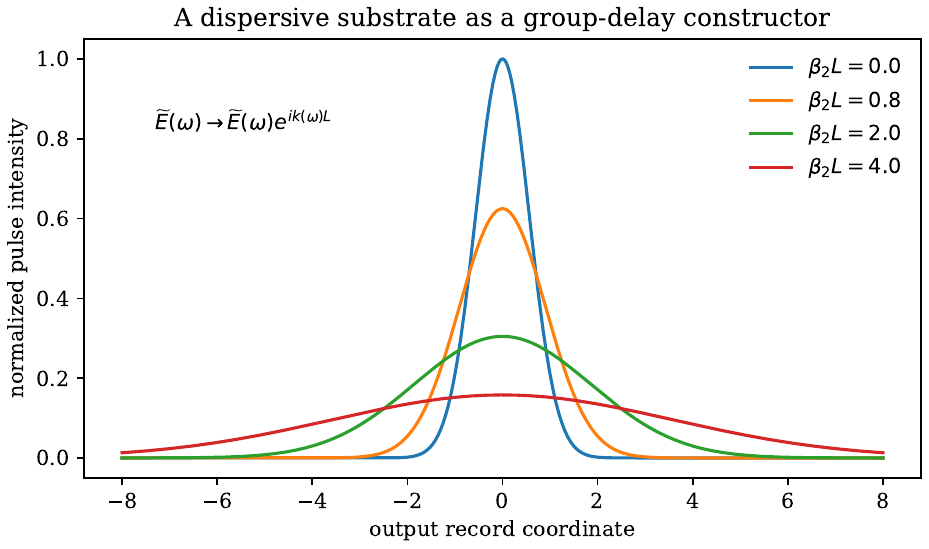}
\caption{A dispersive substrate as a group-delay constructor. The transformation $\widetilde E(\omega)\mapsto \widetilde E(\omega)e^{ik(\omega)L}$ maps spectral attributes into temporal-record attributes. Group delay and pulse broadening acquire meaning through comparison and record-forming tasks.}
\label{fig:dispersion}
\end{figure}

This example also shows why delay is not equivalent to a single scalar in all contexts. A nondispersive delay constructor may approximately preserve the pulse shape, while a dispersive constructor changes the relation between temporal ordering, bandwidth, coherence, and estimation. The task description therefore carries more physical content than a bare parameter $\tau$.

\section{Double-slit diffraction as a phase-delay record task}
\label{sec:double_slit}
The question of whether a double-slit diffraction pattern can be obtained within the present methodology has a simple answer: yes, but with a clarification. The task-based formulation does not replace Maxwell's equations, Huygens-Fresnel diffraction, or Fourier optics. Those theories remain the predictive algorithms that compute the record distribution. The constructor-theoretic contribution is to reinterpret the pattern as the output record of a task that maps path-delay and phase attributes into spatially distinguishable detector records.

Consider two identical slits of width $a$ separated by center-to-center distance $d$, illuminated by a quasi-monochromatic scalar plane wave of wavelength $\lambda$. In the Fraunhofer regime, the ordinary double-slit intensity is
\begin{equation}
I(\theta)=I_0\left(\frac{\sin\beta}{\beta}\right)^2\cos^2\gamma,
\label{eq:double_slit_intensity}
\end{equation}
with
\begin{equation}
\beta=\frac{\pi a}{\lambda}\sin\theta,
\qquad
\gamma=\frac{\pi d}{\lambda}\sin\theta.
\end{equation}
After normalization over the detector domain, this gives the record probability distribution
\begin{equation}
p(\theta|a,d)=\frac{1}{Z(a,d)}\left(\frac{\sin\beta}{\beta}\right)^2\cos^2\gamma.
\label{eq:double_slit_probability}
\end{equation}

The constructor-theoretic reading is obtained by observing that each detection angle corresponds to a path difference
\begin{equation}
\Delta \ell(\theta)=d\sin\theta,
\end{equation}
and hence to a path-delay attribute
\begin{equation}
\Delta\tau(\theta)=\frac{d\sin\theta}{c}.
\end{equation}
Relative to the optical frequency reference $\omega=2\pi c/\lambda$, this delay becomes a phase attribute
\begin{equation}
\phi(\theta)=\omega\Delta\tau(\theta)=\frac{2\pi d}{\lambda}\sin\theta=2\gamma.
\end{equation}
Therefore the interference factor in Eq.~\eqref{eq:double_slit_intensity} is the record signature of a phase-delay comparison task, while the sinc envelope is the record signature of the finite-aperture diffraction task \citep{Goodman2005,BornWolf1999}.

The double-slit constructor may be written schematically as
\begin{equation}
\Ttask_{\rm DS}:x_{\rm slit}(a,d)\otimes x_{\rm field}(\lambda)
\longrightarrow r(\theta),
\end{equation}
with record distribution $p(\theta|a,d)$. The pattern is thus not merely a spatial curve; it is a stable record distribution generated by a transformation that converts aperture geometry, phase-delay equivalence, and finite detector resolution into distinguishable records.

\begin{figure}[t]
\centering
\includegraphics[width=0.84\linewidth]{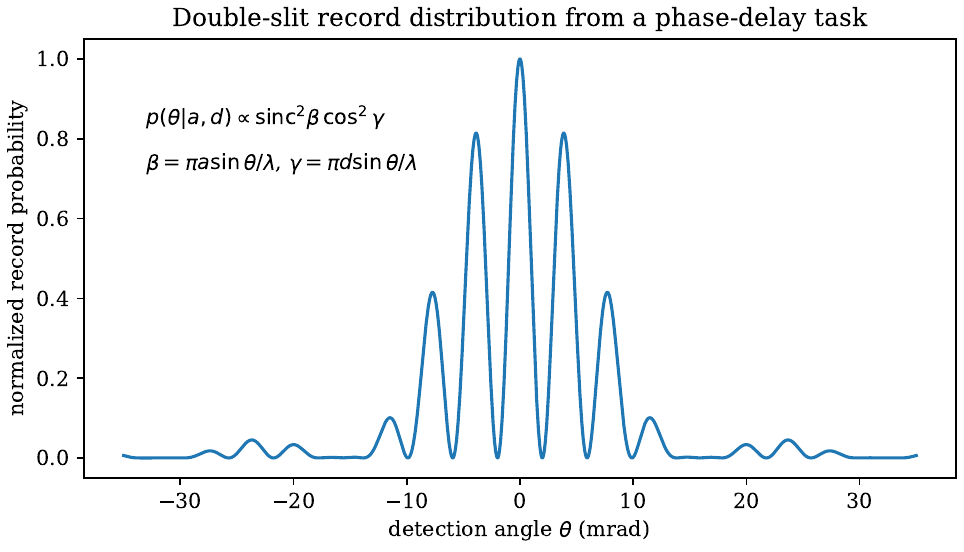}
\caption{Double-slit diffraction as the record distribution of a phase-delay task. The ordinary Fraunhofer pattern is recovered from standard scalar diffraction, while the task-based interpretation identifies the interference fringes as records of path-delay and phase-delay equivalence.}
\label{fig:double_slit}
\end{figure}

\noindent\textit{Comparison with the electromagnetic Fraunhofer result.}
For the same geometry, the standard electromagnetic aperture calculation in the far field gives, for a fixed linear polarization and a thin perfectly transmitting aperture,
\begin{equation}
\widetilde E_{\rm EM}(\theta)\propto
\int_{A_{\rm DS}}\exp\!\big[i k x\sin\theta\big]\dd x,
\qquad k=\frac{2\pi}{\lambda},
\end{equation}
where $A_{\rm DS}$ is the union of the two slit intervals. For slits centered at $x=\pm d/2$, each of width $a$, the integral evaluates to
\begin{equation}
\widetilde E_{\rm EM}(\theta)\propto
2a\frac{\sin\beta}{\beta}\cos\gamma.
\end{equation}
Therefore the electromagnetic intensity is
\begin{equation}
I_{\rm EM}(\theta)\propto
\left(\frac{\sin\beta}{\beta}\right)^2\cos^2\gamma.
\end{equation}
This is identical to Eq.~\eqref{eq:double_slit_intensity} after normalization. The equality does not mean that constructor theory has replaced the electromagnetic calculation. It means that, in the Fraunhofer regime, Maxwellian electromagnetism supplies the possible record distribution and the task formulation assigns it an operational meaning: a double-slit substrate transforms aperture and phase-delay attributes into spatial records.
\begin{figure}[H]
\centering
\includegraphics[width=0.86\linewidth]{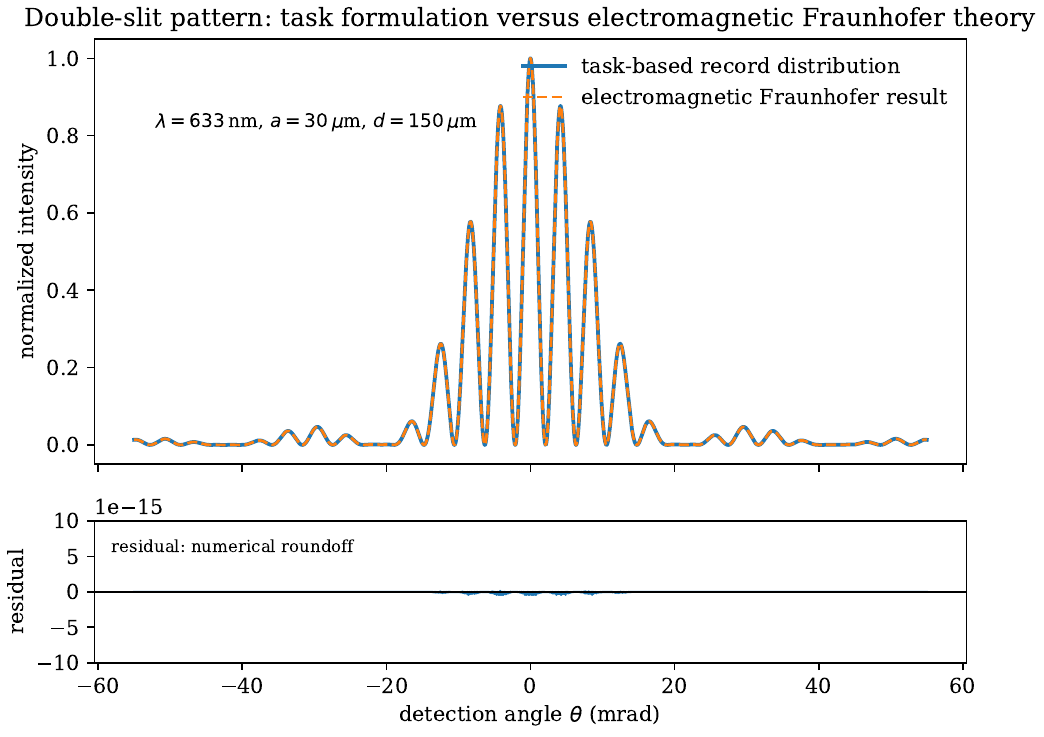}
\caption{Comparison between the constructor-theoretic record distribution and the electromagnetic Fraunhofer aperture result for a double slit with $\lambda=633$~nm, $a=30~\mu$m, and $d=150~\mu$m. In the common scalar, paraxial, thin-aperture regime, both descriptions give the same normalized angular pattern; the residual is numerical roundoff.}
\label{fig:ct_vs_em_double_slit}
\end{figure}

The comparison also identifies the domain of validity. If the slits are subwavelength, metallic, thick, lossy, resonant, or strongly polarization dependent, the electromagnetic prediction must be obtained from the full vector boundary-value problem, for example by mode matching, RCWA, FDTD, or finite-element simulation. In that regime the task formulation remains valid, but $p(\theta|a,d)$ must be computed from the appropriate Maxwell solution rather than from the simple Fraunhofer formula \citep{BornWolf1999,Goodman2005,Stratton1941,Jackson1999}.

This also gives a direct Fisher-information extension. If the slit separation $d$ is unknown, the available information for estimating it from the detected angular records is
\begin{equation}
F(d)=N\int \dd\theta\,\frac{1}{p(\theta|a,d)}\left[\frac{\partial p(\theta|a,d)}{\partial d}\right]^2.
\end{equation}
A task that demands an estimate $\hat d$ with variance below $F(d)^{-1}$ is Fisher-impossible for the specified detector aperture, photon number, wavelength, noise model, and measurement architecture. Similarly, one may define Fisher information for $a$, $\lambda$, an added phase bias, or a temporal path delay. Thus the double-slit pattern can indeed be obtained and used in the new methodology: the ordinary diffraction formula supplies the record distribution, and constructor theory classifies the corresponding estimation and distinguishability tasks as possible or impossible.

\section{Records, irreversibility, and the optical arrow of time}
\label{sec:records}
A unitary optical transformation may be reversible. An ideal delay line has an inverse. A phase shifter can be undone. In contrast, a detector record is usually thermodynamically irreversible: a photon absorption event, an avalanche in a photodiode, or a stored electronic bit is not merely another coherent optical attribute. It is a memory-bearing record.

This distinction matters for optical time. A reversible transformation may encode a delay attribute, but a temporal fact in the laboratory requires a record-forming task,
\begin{equation}
\Rtask:x_F^{(\tau)}\otimes x_C\otimes x_R\longrightarrow r_\tau.
\end{equation}
The arrow-like aspect of measured optical time enters through the production, stabilization, and comparison of records. This is compatible with the broader constructor-theoretic ambition of expressing emergent laws, including thermodynamic irreversibility, as exact statements about possible and impossible tasks \citep{Deutsch2013Constructor,Marletto2021Cannot}.

The present proposal therefore separates three notions that are often conflated: a reversible optical transformation that can be interpreted as a delay; a comparison task that relates the transformation to a reference; and an irreversible record task that stores the temporal distinction. Only their combination yields the ordinary laboratory notion of an observed optical time interval.

\section{Relation to timeless dynamics and constructor theory of time}
\label{sec:relation}
The proposal is close in spirit to relational and timeless approaches, but it is not identical to them. Page and Wootters recover apparent dynamics from correlations between a system and a clock \citep{PageWootters1983}. Rovelli emphasizes relations among variables rather than evolution in an external time \citep{Rovelli2011ForgetTime}. Constructor theory of time asks how duration and dynamics can be meaningful when the fundamental constructor-theoretic laws do not refer to time \citep{DeutschMarlettoTime2025}.

Our optical formulation is more modest. We do not claim to derive the whole temporal structure of physics. We claim that a central part of optical time can be organized as
\begin{equation}
\hbox{optical time}\quad\Longleftrightarrow\quad
\hbox{delay, phase, ordering, synchronization, comparison, and record tasks}.
\end{equation}
This is a domain-specific reconstruction. It shows how the operational content of optical delay, phase, and diffraction-related path differences can be expressed without assigning primitive explanatory status to $t$.

The approach is compatible with standard dynamical optics. Maxwell's equations, Fourier optics, and quantum optics remain the tools that compute the record distributions and the possible/impossible boundaries. The task-based formulation changes the explanatory organization: it asks which transformations make temporal attributes distinguishable and which transformations are forbidden by optical, statistical, or quantum constraints.

\section{Discussion}
\label{sec:discussion}
The constructor-theoretic formulation offers several advantages. First, it clarifies the operational status of optical delay. A delay is not merely a number attached to a path; it is a transformation class defined by comparison tasks and record distinguishability. This matters whenever different physical systems instantiate nominally equivalent delays with different noise, dispersion, loss, or reference dependence.

Second, it gives a foundations-oriented interpretation of Fisher information. Fisher information is not merely a statistical convenience. It quantifies the possibility of transforming a hidden temporal attribute into a stable record. The Cramer-Rao bound is therefore a physical impossibility statement relative to a specified constructor and measurement class.

Third, it unifies phase, delay, synchronization, and diffraction records. Optical phase is not automatically time. It becomes time through a phase-delay equivalence task involving a frequency reference. Synchronization is likewise not primitive simultaneity; it is the possibility of establishing and maintaining a comparison relation among reference substrates. Double-slit interference is a particularly transparent spatial record of phase-delay equivalence, because the path delay $d\sin\theta/c$ is converted into angular fringes.

Fourth, it suggests a classification program for optical time devices. Delay lines, cavities, dispersive media, interferometers, frequency combs, double slits, gratings, and detectors can be classified by the temporal tasks they make possible, the tasks they make impossible, and the resources required to approach the possible/impossible boundary.

There are also limitations. The present work does not provide a complete axiomatization of optical constructor theory. It uses standard classical Fisher information, with quantum Fisher information discussed only as an extension. It treats simple interferometric, dispersive, and double-slit examples rather than complex many-mode quantum optical systems. It also does not settle whether constructor theory is the correct fundamental formulation of physics. Its purpose is narrower: to show that optical time admits a precise task-based reconstruction directly connected to measurable distinguishability.

\section{Outlook}
A quantum version would replace the classical Fisher information by quantum Fisher information and formulate quantum-delay tasks constrained by nonorthogonality, measurement disturbance, and no-cloning \citep{Helstrom1976,BraunsteinCaves1994,Paris2009}. A structured-light version would treat pulse shaping, orbital angular momentum, spatial-mode engineering, and frequency-comb design as constructor engineering for temporal distinguishability. A relativistic version would formulate optical time-transfer and synchronization protocols as tasks constrained by causal structure and gravitational redshift. A thermodynamic version would analyze the energy and entropy costs of producing stable optical time records.

The most direct experimental platform is interferometric delay estimation. One may compare different physical delay constructors--free space, fiber, slow-light media, cavities, programmable pulse shapers, and multi-aperture diffractive systems--by asking not only how much delay they produce, but what temporal tasks they render possible at fixed resources. The experimentally relevant observable is the Fisher-information landscape over visibility, bandwidth, photon number, dispersion, detector jitter, and reference stability.

\section{Conclusion}
\label{sec:conclusion}
We have proposed a constructor-theoretic formulation of optical time. Delay, phase, ordering, synchronization, diffraction path difference, and temporal records were expressed as possible and impossible tasks acting on optical, reference, and record substrates. In this framework, an optical delay is not introduced as primitive time evolution. It is an attribute made meaningful by transformations that generate distinguishable records. Phase becomes temporal only through a reference-dependent phase-delay equivalence task. Synchronization becomes a possible relation among reference substrates. Measurement produces the record structure through which temporal facts enter the laboratory.

The Fisher information provides the quantitative bridge between optical metrology and constructor-theoretic possibility. A temporal-estimation task is Fisher-admissible only when the requested uncertainty is compatible with the available distinguishability resource. Otherwise it is Fisher-impossible. The Cramer-Rao bound can therefore be read as a task impossibility statement. The double-slit diffraction pattern can be recovered within the same framework: standard scalar diffraction computes $p(\theta|a,d)$, while constructor theory interprets that distribution as the record of aperture, path-delay, and phase-comparison tasks.

This viewpoint does not replace ordinary optics. It reorganizes it. Maxwellian and quantum-optical models compute the possible/impossible boundary; constructor theory gives that boundary a foundational interpretation. Optical time, in this sense, is not a primitive background parameter but an operational structure emerging from delay constructors, reference substrates, distinguishable records, and the physical limitations of temporal inference.

\section*{Declarations}
\textbf{Funding.} No external funding is declared for this manuscript draft.\\
\textbf{Conflict of interest.} The authors declare no conflict of interest.\\
\textbf{Data availability.} No experimental data are used. The figure-generation script is included in the accompanying Overleaf project.\\
\textbf{Author contributions.} J.S.-M. conceived the optical formulation and developed the initial theoretical framework, calculations, and figures. O.O.-L. contributed to the refinement of the physical interpretation, foundations-oriented framing, and critical revision of the manuscript. Both authors reviewed and approved the final manuscript draft.

\end{document}